\documentclass[letterpaper,11pt]{article}

\usepackage{amsmath}
\usepackage{amsfonts}
\usepackage{amssymb}
\usepackage{amsthm}
\usepackage[usenames]{color}
\usepackage{fullpage}
\usepackage{url,ifthen}
\usepackage{srcltx}
\usepackage{multirow}
\usepackage{boxedminipage}
\usepackage[margin=1in]{geometry}
\usepackage{nicefrac}
\usepackage{xspace}
\usepackage[noend]{algorithmic}
\usepackage{algorithm}
\usepackage{color}
\definecolor{DarkGreen}{rgb}{0.1,0.5,0.1}
\definecolor{DarkRed}{rgb}{0.5,0.1,0.1}
\definecolor{DarkBlue}{rgb}{0.1,0.1,0.5}
\usepackage[pdftex]{hyperref}
\hypersetup{
    unicode=false,          
    pdftoolbar=true,        
    pdfmenubar=true,        
    pdffitwindow=false,      
    pdftitle={Privately Releasing Conjunctions and the Statistical Query Barrier},    
    pdfauthor={Anupam Gupta, Moritz Hardt, Aaron Roth, Jonathan Ullman}
    pdfsubject={},   
    pdfnewwindow=true,      
    pdfkeywords={keywords}, 
    colorlinks=true,       
    linkcolor=DarkRed,          
    citecolor=DarkGreen,        
    filecolor=DarkRed,      
    urlcolor=DarkBlue,          
}

\def\draft{1}
\def\submit{1}

\ifnum\draft=1 
    \def\ShowAuthNotes{1}
\else
    \def\ShowAuthNotes{0}
\fi

\ifnum\ShowAuthNotes=1
\newcommand{\authnote}[2]{{ \footnotesize \bf{\color{red}[#1's Note: {\color{blue}#2}]}}}
\else
\newcommand{\authnote}[2]{}
\fi

\ifnum\submit=0 
\newcommand{\forsubmit}[1]{#1}
\newcommand{\forreals}[1]{}
\else
\newcommand{\forreals}[1]{#1}
\newcommand{\forsubmit}[1]{}
\fi

\ifnum\draft=1 
\else
\fi

\newtheorem{theorem}{Theorem}[section]
\newtheorem{itheorem}{Informal Theorem}[section]
\newtheorem{corollary}[theorem]{Corollary}
\newtheorem{remark}{Remark}[section]
\newtheorem{lemma}[theorem]{Lemma}
\newtheorem{proposition}{Proposition}[section]
\newtheorem{claim}[theorem]{Claim}
\newtheorem{fact}{Fact}[section]

\theoremstyle{definition}
\newtheorem{definition}{Definition}[section]

\usepackage{times}
\usepackage[varg]{txfonts} 
\renewcommand{\mathbb}{\varmathbb}

\usepackage{microtype}

\newcommand{\Esymb}{\mathbb{E}}
\newcommand{\Psymb}{\mathbb{P}}

\DeclareMathOperator*{\E}{\Esymb}

\DeclareMathOperator*{\ProbOp}{\Psymb r}
\renewcommand{\Pr}{\ProbOp}

\newcommand{\nfrac}{\nicefrac}

\newcommand{\mper}{\,.}
\newcommand{\mcom}{\,,}

\newcommand{\cA}{{\cal A}}

\newcommand{\cD}{{\cal D}}

\newcommand{\cG}{{\cal G}}

\newcommand{\cI}{{\cal I}}

\newcommand{\poly}{{\rm poly}}
\newcommand{\remove}[1]{}

\renewcommand{\epsilon}{\varepsilon}

\newcommand{\Lap}{\mathit{Lap}}

\newcommand{\defeq}{\stackrel{\small \mathrm{def}}{=}}

\newcommand{\KL}{\mathrm{RE}}

\newcommand{\bR}{\mathbb{R}}

\newcommand{\map}{F}
\newcommand{\sens}{\gamma}
\newcommand{\dom}{V}
\newcommand{\rest}{T}
\renewcommand{\path}{P}
\newcommand{\tolf}{\tilde{f}}
\newcommand{\toldom}{\dom'}

\newcommand{\initOneLiners}{%
    \setlength{\itemsep}{0pt}
    \setlength{\parsep }{0pt}
    \setlength{\topsep }{0pt}
}

\title{\bf Privately Releasing Conjunctions\\ and the Statistical Query Barrier}
\author{Anupam Gupta\thanks{Computer Science Department, Carnegie Mellon
    University, Pittsburgh, PA 15213. Research was partly supported by
    NSF award CCF-1016799 and an Alfred P.~Sloan Fellowship.}
\and Moritz Hardt\thanks{Center for Computational Intractability, Department
of Computer Science, Princeton University.  Supported by NSF grants
CCF-0426582 and CCF-0832797. Email: {\tt mhardt@cs.princeton.edu}.}
\and Aaron Roth\thanks{Microsoft Research New England, Cambridge MA. Email: {\tt alroth@cs.cmu.edu}.}
\and Jonathan Ullman\thanks{School of Engineering and Applied Sciences, Harvard University, Cambridge MA.  Supported by NSF grant CNS-0831289.  Email: {\tt jullman@seas.harvard.edu}.}}

\begin{document}
\maketitle

\begin{abstract}

Suppose we would like to know \emph{all}
 answers to a set of statistical queries~$C$
on a data set up to small error, but we can only access the data itself using
statistical queries. A trivial solution is to exhaustively ask all
queries in~$C.$ Can we do any better?

\begin{enumerate}
\item We show that the number of statistical queries necessary and
  sufficient for this task is---up to polynomial factors---equal to the
  agnostic learning complexity of~$C$ in Kearns' statistical query (SQ)
  model. This gives a complete answer to the question when running time
  is not a concern.

\item
We then show that the problem can be solved efficiently (allowing arbitrary
error on a small fraction of queries) whenever the answers to~$C$ can be
described by a submodular function. This includes many natural concept
classes, such as graph cuts and Boolean disjunctions and conjunctions.
\end{enumerate}

While interesting from a learning theoretic point of view, our main
applications are in \emph{privacy-preserving data analysis}:
Here, our second result leads to an algorithm that efficiently
releases differentially private answers to all Boolean conjunctions
with~1\% average error.
This presents significant progress on a key open problem in
privacy-preserving data analysis.
Our first result on the other hand gives unconditional lower bounds on any
differentially private algorithm that admits a (potentially
non-privacy-preserving) implementation using only
statistical queries. Not only our algorithms, but also most known
private algorithms
 can be implemented using only statistical queries, and hence are
constrained by these lower bounds. Our result therefore isolates the
complexity of agnostic learning in the SQ-model as a new barrier in the
design of differentially private algorithms.
\end{abstract}

\vfill
\thispagestyle{empty}
\setcounter{page}{0}
\pagebreak

\section{Introduction}
\label{sec:introduction}

Consider a data set $D\subseteq\{0,1\}^d$ in which each element corresponds to
an individual's record over~$d$ binary attributes. The goal of privacy-preserving
data analysis is to enable rich statistical analyses on the data set while
respecting individual privacy. In paritcular, we would like to guarantee
\emph{differential privacy}~\cite{DMNS06}, a rigorous
notion of privacy that guarantees the outcome of a statistical analysis
is nearly indistinguishable on any two data sets that differ only in a single
individual's data.

One of the most important classes of statistical queries on the data set
are Boolean conjunctions, sometimes called contingency tables or
marginal queries. See, for example, \cite{BCDKMT07,BLR08,KRSU10,UV10}.
A boolean conjunction corresponding to a subset $S\subseteq[d]$ counts
what fraction of the individuals have each attribute in~$S$ set to~$1.$
A major open problem in privacy-preserving data analysis is to
efficiently create a differentially private synopsis of the data set
that accurately encodes answers to all Boolean conjunctions. In this
work we give an algorithm with runtime polynomial in $d$, which outputs
a differentially private data structure that represents all boolean
conjunctions up to an average error of~1\%.

Our result is significantly more general and applies to any collection of
queries that can be described
by a low sensitivity \emph{submodular} function. Submodularity is a property
that often arises in data analysis and machine learning problems
\cite{KG07}, including problems for which privacy is a first-order
design constraint\footnote{For example, Kempe, Kleinberg, and Tardos
  show that for two common models of influence propagation on social
  networks, the function capturing the ``influence'' of a set of users
  (perhaps the targets of a viral marketing campaign) is a monotone
  submodular function \cite{KKT03}.}. Imagine, for example, a social network
on $d$ vertices. A data analyst may wish to analyze the size of the cuts
induced by various subsets of the vertices.  Here, our result provides a data
structure that represents all cuts up to a small average error.
Another important example of submodularity is the \emph{set-coverage}
function, which given a set system over elements in some universe $U$,
represents the number of elements that are covered by the union of any
collection of the sets.

The size of our data structure grows exponentially in the inverse error
desired, and hence we can represent submodular functions only up to
constant error if we want polynomial query complexity. \emph{Can any
  efficient algorithm do even better?} We give evidence that in order to
do better, fundamentally new techniques are needed. Specifically, we
show that no polynomial-time algorithm that guarantees small error for \emph{every}
boolean conjunction can do substantially better if
the algorithm permits an implementation that only accesses the database
through statistical queries.  This statement holds regardless of whether such an
implementation is privacy-preserving. (A statistical query is given by a
function $q\colon\{0,1\}^d\to\{0,1\}$, to which the answer is $\E_{x\in
  D}[q(x)]$.)

We show this limitation using connection between the data release
problem and standard problems in learning theory.  Putting aside privacy concerns,
we pose the following question: \emph{How many statistical queries to a
  data set are necessary and sufficent in order to approximately answer
  \emph{all} queries in a class~$C$?} We show that the number of
statistical queries necessary and sufficient for this task is, up to a
factor of~$O(d)$, equal to the agnostic learning complexity of~$C$ (over
arbitrary distributions) in Kearns' statistical query (SQ)
model~\cite{Kearns98}. Using an SQ lower bound for agnostically
learning monotone conjunctions shown by Feldman~\cite{Feldman10}, this
connection implies that no polynomial-time algorithm operating in the
SQ-model can release even monotone conjunctions to subconstant
error. Since monotone conjunction queries can be described by a submodular function,
the lower bound applies to releasing submodular functions as well.

While the characterization above is independent of privacy concerns, it
has two immediate implications for private data release:
\begin{itemize}
\item Firstly, it also characterizes what can be released in the
  \emph{local privacy} model of Kasiviswanathan et al.~\cite{KLNRS09};
  this follows from the fact that~\cite{KLNRS09} showed that SQ
  algorithms are precisely what can be computed in the local privacy
  model.
\item Secondly, and perhaps even more importantly, it gives us the
  claimed unconditional lower bounds on the running time of any
  query-release algorithm that permits an implementation using only
  statistical queries---regardless of whether its privacy analysis can
  be carried out in the local privacy model.  To our knowledge, this
  class includes almost all privacy preserving algorithms developed to
  date, including the recently introduced Median Mechanism \cite{RR10}
  and Multiplicative Weights Mechanism \cite{HardtRo10}\footnote{A
    notable exception is the private parity-learning algorithm of
    \cite{KLNRS09}, which explicitly escapes the statistical query
    model.}. Note that these mechanisms cannot be implemented in the
  local privacy model while preserving their privacy guarantees, because
  they will have to make too many queries. Indeed, they are capable of
  releasing conjunctions to subconstant error! Yet, they can be
  implemented using only statistical queries, and so our lower bounds
  apply to their running time.
\end{itemize}
To summarize, our results imply that if we want to develop efficient
algorithms to solve the query release problem for classes as expressive
as monotone conjunctions (itself an extremely simple class!), we need to
develop techniques that are able to sidestep this \emph{statistical query
barrier.} On a conceptual level, our results present new reductions from
problems in differential privacy to problems in learning theory.

\subsection{Overview of our results}

In this section we give an informal statement of our theorems with pointers to
the relevant sections. Our theorem on approximating submodular
functions is proved in Section~\ref{sec:submodular}. The definition of
submodularity is found in the Preliminaries (Section~\ref{sec:prelim}).

\begin{itheorem}[Approximating submodular functions]
Let $\alpha>0,\beta>0.$ Let $f\colon \{0,1\}^d\to[0,1]$ be a submodular
function.  Then, there is an algorithm with runtime
$d^{O(\log(1/\beta)/\alpha^2)}$ which produces an
approximation~$h\colon\{0,1\}^d\to[0,1]$ such that
$\Pr_{x\in\{0,1\}^d}\{|f(x)-h(x)|\le\alpha\}\ge1-\beta.$
\end{itheorem}

In Section~\ref{sec:conjunctions} we then show how this algorithm gives the
following differentially private release mechanism for Boolean conjunctions.
The definition of differential privacy is given in Section~\ref{sec:prelim}.

\begin{itheorem}[Differentially private query release for conjunctions]
Let $\alpha>0,\beta>0.$ There is an $\epsilon$-differentially private
algorithm with runtime $d^{O(\log(1/\beta)/\alpha^2)}$
which releases the set of Boolean conjunctions with error at most
$\alpha$ on a $1-\beta$ fraction of the queries provided that
$|D|\ge d^{O(\log(1/\beta)/\alpha^2)}/\epsilon\mper$
\end{itheorem}

The guarantee in our theorem can be refined to give an $\alpha$-approximation
to a $1-\beta$ fraction of the set of $w$-way conjunctions (conjunctions of
width~$w$) for all $w\in\{1,...,d\}.$ Nevertheless, our algorithm has the
property that the error may be larger than~$\alpha$ on a small fraction of the
queries. We note, however, that for $\beta\le\alpha^p/2$ our guarantee is
stronger than error~$\alpha$ in the $L_p$-norm which is also a natural
objective that has been considered in other works.
For example, Hardt and Talwar study
error bounds on mechanisms with respect to the Euclidean norm across all
answers~\cite{HardtTa10}.
From a practical
point of view, it also turns out that some privacy-preserving algorithms in the
literature indeed only require the ability to answer \emph{random} conjunction
queries privately, e.g., \cite{JPW09}.

Finally, in Section~\ref{sec:agnostic}, we study the general query release
problem and relate it to the agnostic learning complexity in the Statistical
Query model.
\begin{itheorem}[Equivalence between query release and agnostic learning]
Suppose there exists an algorithm that learns a class~$C$ up to error $\alpha$
under arbitrary distributions using at most~$q$ statistical queries. Then,
there is a release mechanism for~$C$ that makes at most~$O(qd/\alpha^2)$
statistical queries.

Moreover, any release mechanism for $C$ that makes at most $q$ statistical
queries implies an agnostic learner that makes at most $2q$ queries.
\end{itheorem}
While both reductions preserve the query complexity of the
problem neither reduction preserves runtime.
We also note that our equivalence characterization is more general than what we
stated: the same proof shows that agnostic learning of a class $C$ is (up
to small factors) information theoretically equivalent to releasing the
answers to all queries in a class $C$ for any class of algorithms that
may access the database only in some restricted manner. The ability to
make only SQ queries is one restriction, and the requirement to be
differentially private is another. Thus, we also show that on a class by
class basis, the privacy cost of releasing the answers to a class of
queries using any technique is not much larger than the privacy cost of
simply optimizing over the same class to find the query with the highest
value, and vice versa.

\paragraph{Our techniques.}

Our release algorithm is based on a structural theorem about general submodular
functions $f:2^U\rightarrow [0,1]$ that may be of independent
interest. Informally, we show that any submodular function has a
``small'' ``approximate'' representation. Specifically, we show that for
any $\alpha > 0$, there exist at most $|U|^{2/\alpha}$ submodular
functions $g_i$ such that each $g_i$ satisfies a strong Lipschitz
condition, and for each $S \subset U$, there exists an $i$ such that
$f(S) = g_i(S)$.  We then take advantage of Vondrak's observation in~\cite{Von10} that
Lipschitz submodular functions are \emph{self-bounding},
which allows us to apply recent dimension-free concentration bounds for
self-bounding functions~\cite{BLM00, BLM09}. These concentration results imply
that if we associate each function $g_i$ with its expectation, and
respond to queries $f(S)$ with $\E[g_i(S)]$ for the appropriate $g_i$,
then most queries are answered to within only $\alpha$ additive
error. This yields an algorithm for \emph{learning} submodular
functions over product distributions, which can easily be made privacy
preserving when the values $f(S)$ correspond to queries on a sensitive database.

Our characterization of the query complexity of the release problem in
the SQ model uses the multiplicative weights
method~\cite{LittlestoneWa94,AroraHaKa05} similar to how it was
used recently in~\cite{HardtRo10}. That is we maintain a distribution over the
universe on which the queries are defined. What is new is the observation
that an agnostic learning algorithm for a class~$C$ can be used to find
a query from~$C$ that distinguishes between the true data set and our
distribution as much as possible.
Such a query can then be used in the multiplicative weights
update to reduce the relative entropy between the true data set and our
distribution significantly. Since the relative entropy is nonnegative there
can only be a few such steps before we find a distribution which provides a
good approximation to the true data set on \emph{all} queries in the
class~$C.$

\subsection{Related Work}

\paragraph{Learning Submodular Functions.}
The problem of learning submodular functions was recently introduced by
Balcan and Harvey \cite{BH10}; their PAC-style definition was different
from previously studied point-wise learning
approaches~\cite{GHIM09,FS08}. For product distributions, Balcan and
Harvey give an algorithm for learning monotone, Lipschitz continuous
submodular functions up to constant \emph{multiplicative} error using
only random examples. \cite{BH10} also give strong lower bounds and
matching algorithmic results for non-product distributions. Our main
algorithmic result is similar in spirit, and is inspired by their
concentration-of-measure approach. Our model is different from theirs,
which makes our results incomparable. We introduce a decomposition that
allows us to learn arbitrary (i.e. potentially non-Lipschitz,
non-monotone) submodular functions to constant \emph{additive}
error. Moreover, our decomposition makes value queries to the submodular
function, which are prohibited in the model studied by \cite{BH10}.

\forsubmit{\vspace{-0.15in}}
\paragraph{Information Theoretic Characterizations in Privacy.}
Kasiviswanathan et al. \cite{KLNRS09} introduced the \emph{centralized} and
\emph{local} models of privacy and gave information theoretic
characterizations for which classes of functions could be \emph{learned} in
these models: they showed that information theoretically, the class of
functions that can be learned in the centralized model of privacy is
equivalent to the class of functions that can be agnostically PAC learned, and
the class of functions that can be learned in the local privacy model is
equivalent to the class of functions that can be learned in the SQ model of
Kearns \cite{Kearns98}.

Blum, Ligett, and Roth \cite{BLR08} considered the \emph{query release}
problem (the task of releasing the approximate value of all functions in some
class) and characterized exactly which classes of functions can be information
theoretically released while preserving differential privacy in the \emph{centralized} model of data privacy. They also posed
the question: which classes of functions can be released using mechanisms that
have running time only polylogarithmic in the size of the data universe and the
class of interest? In particular, they asked if conjunctions were such a
class.

In this paper, we give an exact information theoretic characterization of
which classes of functions can be released in the SQ model, and hence in the
local privacy model: we show that it is exactly the class of functions that
can be \emph{agnostically learned} in the SQ model. We note that the agnostic SQ learnability of a class $C$ (and hence, by our result, the SQ releasability of $C$) can also be characterized by combinatorial properties of $C$, as done by Blum et al. \cite{BFJ+94} and recently Feldman \cite{Feldman10}.

\forsubmit{\vspace{-0.15in}}
\paragraph{Lower bounds and hardness results.}
There are also several conditional lower bounds on the running time of private
mechanisms for solving the query release problem. Dwork et al. \cite{DNRRV09}
showed that under cryptographic assumptions, there exists a class of queries
that can be privately released using the inefficient mechanism of \cite{BLR08}, but cannot be
privately released by any mechanism that runs in time polynomial in the dimension of the data
universe (e.g.~$d$, when the data universe is $\{0,1\}^d$).
Ullman and Vadhan \cite{UV10} extended this result to the class of
conjunctions: they showed that under cryptographic assumptions, no polynomial
time mechanism \emph{that outputs a data set} can answer even the set of
$d^2$ conjunctions of two-literals!

The latter lower bound applies only to the class of mechanisms that
output data sets, rather than some other data structure encoding their answers, and only to mechanisms
that answer \emph{all} conjunctions of two-literals with small error. In fact, because there are only
$d^2$ conjunctions of size 2 in total, the hardness result of \cite{UV10} does
not hold if the mechanism is allowed to output some other data structure --
such a mechanism can simply privately query each of the $d^2$ questions.

We circumvent the hardness result of \cite{UV10} by outputting a data
structure rather than a synthetic data set, and by releasing all conjunctions with small \emph{average} error.
Although there are no known computational lower bounds for releasing conjunctions with small average error,
even for algorithms that output a data set, since our algorithm does not output a data set, our approach may be useful in circumventing the lower bounds of \cite{UV10}.

We also prove a new unconditional (information theoretic) lower
bound on algorithms for privately releasing monotone conjunctions that applies
to the class of algorithms that interact with the data using only SQ queries:
no such polynomial time algorithm can release monotone conjunctions with
$o(1)$ average error. We note that our lower bound does not depend on the
output representation of the algorithm. Because almost all known private
algorithms can indeed be implemented using statistical queries, this provides
a new perspective on sources of hardness for private query release.
We note that information theoretic lower bounds on the query complexity imply
lower bounds on the running time of such differentially private algorithms.

There are also many lower bounds on the \emph{error} that must be introduced by any private mechanism, 
independent of its running time.   In particular, Kasiviswanathan 
et.~al.~\cite{KRSU10} showed that average error of $\Omega(1/\sqrt{n})$ is necessary for private mechanisms that answer all conjunction queries of constant size.  Recently, this work was extended by De~\cite{De11} to apply to mechanisms that are allowed to have arbitrarily large error on a constant fraction of conjunction queries of constant size.  These results extend earlier results by Dinur and Nissim~\cite{DinurN03} showing that average error $\Omega(1/\sqrt{n})$ is necessary for random queries.

\forsubmit{\vspace{-0.15in}}
\paragraph{Interactive private query release mechanisms.}
Recently, Roth and Roughgarden \cite{RR10} and Hardt and Rothblum
\cite{HardtRo10} gave interactive private query release mechanisms that allow
a data analyst to ask a large number of questions, while only expending their
privacy budgets slowly. Their privacy analyses depend on the fact that only a
small fraction of the queries asked necessitate updating the internal state of
the algorithm. However, to answer large classes of queries, these algorithms
need to make a large number of statistical queries to the database, even
though only a small number of statistical queries result in update steps!
Intuitively, our characterization of the query complexity of the release
problem in the SQ model is based on two observations: first, that it would be
possible to implement these interactive mechanisms using only a small number
of statistical queries if the data analyst was able to ask only those queries
that would result in update steps, and second, that finding queries that
induce large update steps is exactly the problem of agnostic learning.

\section{Preliminaries}
\label{sec:prelim}

\paragraph{Differential privacy and counting queries.}
We study the question of answering \emph{counting queries} over a
database while preserving differential privacy. Given an arbitrary domain $X$,
we consider databases $D \in X^n$.  We write $n = |D|$.  Two databases $D = (x_1, \dots, x_n)$ and
$D' = (x'_1, \dots, x'_n)$ are called \emph{adjacent} if they differ only in one entry.  That is, there exists $i \in [n]$ such that for every $j \neq i$, $x_j = x'_j$.
We are interested in algorithms (or \emph{mechanisms}) that map databases to
some abstract range $\mathcal{R}$ while satisfying $\epsilon$-differential
privacy:
\begin{definition}[Differential Privacy \cite{DMNS06}]
A mechanism $\mathcal{M}:X^*\rightarrow \mathcal{R}$ satisfies
$\epsilon$-differential privacy if for all $S \subset \mathcal{R}$ and every
pair of two adjacent databases $D,D',$ we have
$\Pr(\mathcal{M}(D) \in S)\leq e^\epsilon \Pr(\mathcal{M}(D') \in S)
\mper$
\end{definition}
A \emph{counting query} is specified by a predicate~$q\colon X\rightarrow
[0,1]$. We will denote the answer to a counting query (with some abuse of
notation) by~$q(D) = \frac1n\sum_{x \in D}q(X)\mper$
Note that a count query can differ by at most $1/n$ on any two adjacent
databases. In particular, adding Laplacian noise of magnitude $1/\epsilon n,$
denoted $\Lap(1/\epsilon n),$ guarantees $\epsilon$-differential privacy on
a single count query (see \cite{DMNS06} for details).

\forsubmit{\vspace{-0.15in}}
\paragraph{The statistical query model and its connection to differential
privacy.}
We will state our algorithms in Kearns' statistical query (SQ) model. In this
model an algorithm~$A^{\cal O}$ can access a distribution~$D$ over a
universe~$X$ only through \emph{statistical queries} to an oracle~${\cal O}.$
That is, the algorithm may ask any query~$q\colon X\to[0,1]$ and the
oracle may respond with any answer~$a$ satisfying $|a-\E_{x\sim
D}q(x)|\le\tau\mper$ Here, $\tau$ is a parameter called the \emph{tolerance}
of the query.

In the context of differential privacy, the distribution~$D$ will typically be
the uniform distribution over a data set of size~$n.$ A statistical query is
then just the same as a counting query as defined earlier. Since SQ algorithms
are tolerant to noise it is not difficult to turn them into differentially
private algorithms using a suitable oracle. This observation is not new, and has been used previously, for example by Blum et al. \cite{BDMN05} and Kasiviswanathan et al. \cite{KLNRS09}.
\begin{proposition}\label{prop:sq2dp}
Let $A$ denote an algorithm that requires $k$ queries of tolerance~$\tau.$
Let ${\cal O}$ denote the oracle that outputs $\E_{x\sim
D}q(x)+\Lap(k/n\epsilon).$ Then, the algorithm $A^{\cal O}$ satisfies
$\epsilon$-differential privacy and with probability at least
$1-\beta,$ the oracle answers all $q$ queries with error at most~$\tau$
provided that
$n\ge \frac{k(\log k + \log(1/\beta))}{\epsilon\tau}\mper$
\end{proposition}
\begin{proof}
The first claim follows directly from the properties of the Laplacian
mechanism and the composition property of $\epsilon$-differential privacy.
To argue the second claim note that $\Pr(|\Lap(\sigma)|\ge \tau)\le
\exp(-\tau/\sigma)\mper$ Using that $\sigma = k/n\epsilon$
and the assumption on~$n$, we get that this probability is less
than~$\beta/k.$ The claim now follows by taking a union bound over all $k$~queries.
\end{proof}

\paragraph{Query release.}
A \emph{concept class} (or \emph{query class}) is a set of predicates from $X \to [0,1]$.
\begin{definition}[Query Release]
Let $C$ be a concept class. We say that an algorithm~$A$
\emph{$(\alpha,\beta)$-releases}~$C$ over a data set~$D$ if
$\Pr_{q\sim C}\{|q(D)-A(q)|\le\alpha\}\ge1-\beta\mper$
\end{definition}
Specifically, we are interested in algorithms which release $C$ using few
statistical queries to the underlying data set.
We will study the query release problem by considering the function $f(q) = q(D)$.
In this setting, releasing a concept class $C$ is equivalent to \emph{approximating} the function
$q$ is the following sense

\begin{definition}
We say that an algorithm $A$ \emph{$(\alpha,\beta)$-approximates} a function $f\colon
2^U\to[0,1]$ over a distribution $\cD$ if
$\Pr_{S\sim \cD}\{|f(S)-A(S)|\le\alpha\}\ge1-\beta\mper$
\end{definition}

For many concept classes of interest, the function $f(q) = q(D)$ will be \emph{submodular}, defined
next.

\paragraph{Submodularity.}
  Given a universe $U$, a function $f:2^U\rightarrow\bR$ is called
  \emph{submodular} if for all $S, T \subset U$ it holds that
    $f(S \cup T) + f(S \cap T) \leq f(S) + f(T)\mper$
  We define the \emph{marginal value} of $x$ (or \emph{discrete
    derivative}) at $S$ as $\partial_x f(S)=f(S\cup\{x\})-f(S).$
\begin{fact}\label{fact:marginal}
  A function $f$ is submodular if and only if
    $\partial_x f(S)\ge\partial_x f(T)$
for all $S\subseteq T\subseteq U$ and all $x\in U.$
\end{fact}

\begin{definition} A function $f: 2^U \to \mathbb{R}$ is $\sens$-Lipschitz if for every $S \subseteq U$ and $x \in U$, $|\partial_x f(S)| \leq \sens$.
\end{definition}

\paragraph{Concentration bounds for submodular functions.}

The next lemma was shown by Vondrak~\cite{Von10} building on concentration
bounds for so-called self-bounding functions due to~\cite{BLM00,BLM09}.

\begin{lemma}[Concentration for submodular functions]
\label{lem:concentration-submodular}
Let $f\colon2^{U}\rightarrow \bR$ be a $1$-Lipschitz submodular function.
Then for any product distribution~$\cD$ over $2^U,$ we have
\begin{equation}
\Pr_{S \sim \cD}\left\{\left|f(S) - \E f(S)\right| \geq t\right\}
\leq 2\exp\left(-\frac{t^2}{2(\E f(S) + 5t/6)}\right),
\end{equation}
where the expectations are taken over $S \sim \cD$.
\end{lemma}

We obtain as a simple corollary

\begin{corollary}
\label{cor:concentration-submodular}
Let $f\colon2^{U} \rightarrow [0,1]$ be a $\sens$-Lipschitz submodular function.  Then for any product distribution~$\cD$ over $2^U,$ we have
\begin{equation}
\Pr_{S \sim \cD}\left\{\left|f(S) - \E f(S)\right| \geq \sens t\right\}
\leq 2\exp\left(-\frac{t^2}{2(1/\sens + 5t/6)}\right),
\end{equation}
where the expectations are taken over $S \sim \cD$.
\end{corollary}

\section{Approximating Submodular Functions}
\label{sec:submodular}

Our algorithm for approximating submodular functions is based on a
structural theorem, together with some strong concentration inequalities
for submodular functions (see Lemma~\ref{lem:concentration-submodular}). 
The structure theorem essentially says that we can decompose any bounded submodular function
into a small collection of Lipschitz submodular functions, one for each region of the domain.
In this section, we prove our structure theorem, present our algorithm, and
prove its correctness.

\subsection{Monotone Submodular Functions}
We begin with a simpler version of the structure theorem.  This version will
be sufficient for approximating bounded monotone submodular functions from
value queries, and will be the main building block in our stronger results,
which will allow us to approximate arbitrary bounded submodular functions,
even from ``tolerant'' value queries.

Our structure theorem follows from an algorithm that decomposes a given
submodular function into Lipschitz submodular functions. The algorithm is
presented next and analyzed in Lemma~\ref{lem:structure}.

\begin{algorithm}
\caption{Decomposition for monotone submodular functions}
\textbf{Input:} Oracle access to a submodular function $f\colon 2^U\to[0,1]$ and parameter
$\sens > 0.$
  \begin{algorithmic}
    \STATE{\textbf{Let} $\prec$ denote an arbitrary ordering of $U$.}
    \STATE{\textbf{Let} $\mathcal{I} \leftarrow \{\emptyset\}$}
    \FOR{$x \in U$ (in ascending order under $\prec$)}
    \STATE{$\cI' \leftarrow \emptyset$}
    \FOR{$B \in \cI$}
    \STATE{\textbf{if} $\partial_{x} f(B) > \sens$ \textbf{then}
      $\cI' \leftarrow \cI' \cup \{ B \cup \{x\} \}$}
    \ENDFOR
    \STATE{$\cI \leftarrow \cI \cup \cI'$}
    \ENDFOR
    \STATE{\textbf{Let}
$\dom(S) = \{x \in U \mid \partial_x f(S) \leq \sens\}$
denote the set of elements that have small
marginal value with respect to $S\subseteq U.$}
  \end{algorithmic}

\textbf{Output:} the 
collection of functions $\cG = \{g^B \mid B \in \cI\},$ where 
for $B \in \cI$ we define the
function $g^{B}: 2^{\dom(B)} \to [0,1]$ as $g^{B}(S) = f(S \cup B).$
\label{alg:decompose}
\end{algorithm}
  
\begin{lemma}
\label{lem:structure}
Given any submodular function $f\colon2^U\to[0,1]$ and $\sens > 0,$
Algorithm~\ref{alg:decompose} makes the following guarantee. There are 
maps $\map,\rest\colon 2^U\to 2^U$ such that:
\begin{enumerate}
\item\emph{(Lipschitz)}  \label{item:Lipschitz} For every $g^B \in \cG$, $g^B$ is submodular and satisfies $\sup_{x \in \dom(B), S \subseteq \dom(B)} \partial_x g^{B}(S) \leq \sens$.
\item\emph{(Completeness)} \label{item:mapping} For every 
$S \subseteq U$, $\map(S) \subseteq S \subseteq \dom(\map(S))$
and $g^{\map(S)}(S) = f(S).$
\item\emph{(Uniqueness)} \label{item:unique} 
For every $S\subseteq U$ and every $B \in \cI$, we have $\map(S) = B$ if and only if $B \subseteq S \subseteq \dom(B)$ and $S \cap \rest(B) = \emptyset$. 
\item\emph{(Size)} \label{item:efficient} The size of $\cG$ is at most $|\cG| =
|U|^{O(1/\sens)}$.  Moreover, given oracle access to $f,$ 
we 
compute $\map, \dom, \rest$ in time $|U|^{O(1/\sens)}$.
  \end{enumerate}
\end{lemma}

Note that the lemma applies to non-monotone submodular functions $f$ as well;
however, since our release algorithm will require the stronger condition
$\sup_{x \in \dom(B), S \subseteq \dom(B)} |\partial_x g(S)| \le \sens$, the
lemma will only be sufficient for releasing monotone submodular functions
(where it holds that $|\partial_xg(S)| \le \sens \iff \partial_xg(S) \le
\sens$). We will return to the non-monotone case later.

\begin{proof}
Algorithm~\ref{alg:decompose} always terminates and we have
the following bound on the size of $\mathcal{I}.$
\begin{claim} \label{clm:setsize}
  $|\mathcal{I}| \leq |U|^{1/\sens}$
\end{claim}
  \begin{proof} Let $B \in \mathcal{I}$ be a set, $B = \{x_1, \dots, x_{|B|}\}$.  Let $B_0 = \emptyset$
  and $B_i = \{x_1, \dots, x_i\}$ for $i = 1,\dots,|B|-1$.  Then
  \begin{equation}
    1 \geq f(B) = \sum_{i = 0}^{|B|-1} \partial_{x_{i+1}} f(B_{i}) > |B| \cdot \sens \mper
  \end{equation}
  Therefore, it must be that $|B| \leq 1/\sens$, and there are at
  most $|U|^{1/\sens}$ such sets over $|U|$ elements.
  \end{proof}
Item~\ref{item:Lipschitz} is shown next.
\begin{claim}[Lipschitz]\label{clm:lipschitz}
   For every $g^{B} \in \cG$, $g^B$ is submodular and $\sup_{x\in \dom(B), S\subseteq \dom(B)} \partial_x g^{B}(S)\le\sens$.
\end{claim}
\begin{proof}
    Submodularity follows from the fact that $g^B$ is a ``shifted'' version of $f$.  Specifically, if $T \subseteq S$,
    then $\partial_x g^{B}(S) = \partial_x f(B \cup S) \leq \partial_x f(B \cup T) = \partial_x g^{B}(T)$, where the inequality is by submodularity of $f$.  
    
    To establish the Lipschitz property, we note that by the definition of $\dom$, $\partial_x f(B) \leq \sens$ for every $x \in \dom(B)$.  Also, by the submodularity of $f$, we have $\partial_x g^{B}(S) = \partial_x f(B \cup S) \leq \partial_x f(B) \leq \sens$.
  \end{proof}

\paragraph{Definition of $\map$ and proof of Item~\ref{item:mapping}.}
Now we turn to constructing the promised mappings $\map$ and $\rest$ in
order to Properties~\ref{item:mapping} and~\ref{item:unique}.  Roughly, we
want $\map(S)$ to choose a maximal set in $\cI$ such that $\map(S) \subseteq
S$, in order to assure that $S \subseteq \dom(\map(S))$.  This task is
complicated by the fact that there could be many such sets.  We want to be
able to choose a unique such set, and moreover, given any such set $B$,
determine efficiently if $\map(S) = B$.  To achieve the former task, we define
a specific, deterministic mapping $\map(S)$ and to achieve the latter we will
carefully define the mapping $\rest$.

We define $\map(S)$ as follows:
  \begin{algorithmic}
    \STATE{\textbf{let} $j \leftarrow 0$, $B_j \leftarrow \emptyset$ }
    \FOR{$x \in U$ (in ascending order under $\prec$)}
    \STATE{\textbf{if} $x \not\in \dom(B_j)$ and $x \in S$ \textbf{then}
      $B_{j+1} \leftarrow B_{j} \cup \{x\}$, $j \leftarrow j+1$}
    \ENDFOR
    \RETURN{$\map(S) = B_j$.}
  \end{algorithmic}
  Note that this procedure is similar to the procedure we use to construct
$\cI$.  To construct $\cI$, we gradually constructed a tree of sets, where
each set $B \in \cI$ had a child for every set $B \cup \{x\}$ such that $x$
has high influence on $B$ ($x \not\in \dom(B)$..  The procedure $\map(S)$
differs in that it only constructs a single root-leaf path in this tree, where
for each $B_j$ in the path, the next set in the path is $B_j \cup \{x\}$ where
$x$ is the \emph{minimal} $x \in S$ that has high influence on $B_j$ (and has
not already been considered by $\map(S)$.  We will use $\path(S) = (B_0
\subset B_1 \subset \dots \subset \map(S))$ to denote this path, which is the
sequence of intermediate sets $B_j$ in the execution of $\map(S)$.  Given
these observations, we can state the following useful facts about $\map$.

\begin{fact} \label{fact:uniquepath}
If $\map(S) = B$, then $\path(S) = \path(B)$.  Moreover, for every $S \in U$, $\path(S) \subseteq \cI$.
\end{fact}

We can now establish Property~\ref{item:mapping} by the following claim.
 \begin{claim}[Completeness]
\label{clm:complete}
  For every $S \subseteq U$, $\map(S) \subseteq S \subseteq \dom(\map(S))$, and $g^{\map(S)}(S) = f(S)$.
  \end{claim}
\begin{proof}
Let $\path(S) = B_0 \subset B_1 \subset \dots \subset \map(S)$.  $\map(S)$
always checks that $x \in S$ before including an element $x$, so $\map(S)
\subseteq S$.  To see that $S \subseteq \dom(\map(S))$, assume there exists $x
\in S \setminus \dom(\map(S))$.  By submodularity we have $\partial_x f(B_j)
\geq \partial_x f(\map(S)) > \sens$ for every set $B_j$.  But if $\partial_x
f(B_j) > \sens$ for every $B_j$ and $x \in S$, it must be that $x \in
\map(S)$.  But then $\partial_x f(\map(S)) = 0$, contradicting the fact that
$x \not\in \dom(\map(S))$.
  
  Finally, we note that since $S \subseteq \dom(\map(S))$, $g^{\map(S)}(S)$ is
defined ($S$ is in the domain of $g^{\map(S)}$) and since $\map(S) \subseteq
S$, $g^{\map(S)}(S) = f(\map(S) \cup S) = f(S)$.  
\end{proof}
  
\paragraph{Definition of $T$ and proof of Item~\ref{item:unique}.}
We will now define the mapping $\rest$. The idea is to consider a set $B \in \cI$ and $\path(B)$ and consider all the elements we had to ``reject'' on the way from the root to $B$.  We say that an element $x \in U$ is ``rejected'' if, when $x$ is considered by $\map(S)$, it has high influence on the current set, but is not in $B$.  Since any set $S$ such that $B = \map(S)$ satisfies $\path(S) = \path(B)$ (Fact~\ref{fact:uniquepath}), and any set $S$ that contains a rejected element would have taken a different path, we will get that the elements $x \in \rest(B)$ ``witness'' the fact that $B \neq \map(S)$.  We define the map $\rest(B)$ as follows:
 \begin{algorithmic}
    \STATE{\textbf{let} $j \leftarrow 0$, $B_j \leftarrow \emptyset$, $R \leftarrow \emptyset$ }
    \FOR{$x \in U$ (in ascending order under $\prec$)}
    \STATE{\textbf{if} $x \not\in \dom(B_j)$ and $x \not\in B$ \textbf{then} $R \leftarrow R \cup \{x\}$}
    \STATE{\textbf{else if} $x \not\in V_{B_{j}}$ and $x \in B$ \textbf{then}
      $B_{j+1} \leftarrow B_{j} \cup \{x\}$, $j \leftarrow j+1$}
    \ENDFOR
    \RETURN{$\rest(B) = R$.}
  \end{algorithmic}
  
  We'll establish Property~\ref{item:unique} via the following two claims.
  
  \begin{claim}\label{clm:unique1}
  If $B = \map(S)$, then $B \subseteq S \subseteq \dom(B)$ and $S \cap \rest(B) = \emptyset$.
  \end{claim}
  \begin{proof}
  	We have already demonstrated the first part of the claim in Claim~\ref{clm:complete}, so we focus on the claim that $S \cap \rest(B) = \emptyset$.
  	By Fact~\ref{fact:uniquepath}, every set $S$ s.t. $B = \map(S)$ satisfies $\path(S) = \path(B)$.  Let $(B_0 \subset B_1 \subset \dots \subset B) = \path(B)$.  Suppose there is an element $x \in S \cap \rest(B)$.  Then there is a set $B_j$ such that $x \not\in \dom(B_j)$ and $x \not\in B$.  But since $x \not\in \dom(B_j)$ and $x \in S$, it must be that $x \in B_{j+1}$, contradicting the fact that $B_{j+1} \subseteq B$.
  \end{proof}
  
  Now we establish the converse.
  \begin{claim}\label{clm:unique2}
If $B \subseteq S \subseteq \dom(B)$, $S \cap \rest(B) = \emptyset$, then $B = \map(S)$.
  \end{claim}
  \begin{proof}
  Suppose for the sake of contradiction that there a set $B' \neq B$ such that $B' = \map(S)$.  There exists an element $x \in B \triangle B'$, and we consider the minimal such $x$ under $\prec$.  Let $\path(B) = (B_0 \subset B_1 \subset \dots \subset B)$ and $\path(S) = \path(B') = (B'_0 \subset B'_1 \subset \dots \subset B')$.  Since $x$ is minimal in $B \triangle B'$, there must be $j$ be such that $B_i = B'_i$ for all $i \leq j$, but $x \in B_{j+1} \triangle B'_{j+1}$.  Consider two cases:
  \begin{enumerate}
  \item $B \supset B'$.  Thus $x \in B \setminus B'$  Moreover, since $x \in B \subseteq S$, it must be that when $x$ was considered in the execution of $\map(S)$, and $B'_{j}$ was the current set, it was the case that $x \in \dom(B'_j)$.  But $B_j = B'_j$, so $x \in \dom(B_j)$, contradicting the fact that $x \in B_{j+1}$.

  \item $B \not\supset B'$.  Thus $x \in B' \setminus B$.  Since $x \in B' = \map(S) \subseteq S$ (Claim~\ref{clm:complete}), we have $x \in S$.  Moreover, since $x \in B'_{j+1}$ we must have $x \not\in \dom(B'_{j}) = \dom(B_{j})$.  Thus we have $x \not\in \dom(B_{j})$ and $x \not\in B$, which implies $x \in \rest(B)$, by construction.  Thus $S \cap \rest(B) \neq \emptyset$, a contradiction.
  \end{enumerate} 
  \end{proof}
The previous two claims establish Item~\ref{item:unique}.

Finally we observe that the enumeration of $\cI$ requires time at most $|U|
\cdot |\cI| = |U|^{O(1/\sens)}$, since we iterate over each element of $U$ and
then iterate over each set currently in $\cI$.  We also note that we can
compute the mappings $\map$ and $\rest$ in time linear in $|\cI| =
|U|^{O(1/\sens)}$ and can compute $\dom(B)$ in time linear in $|U|$.  These
observations establish Property~\ref{item:efficient} and complete the proof of
Lemma~\ref{lem:structure}.
\end{proof}

\begin{algorithm}
\caption{Decomposition for monotone submodular functions from tolerant queries}
\textbf{Input:} Tolerant oracle access to a submodular function $f\colon 2^U\to[0,1]$ with tolerance at most $\sens/12$ and parameter $\sens > 0.$
  \begin{algorithmic}
    \STATE{\textbf{Let} $\tolf$ denote the function specified by tolerant oracle queries to $f$ such that for every $S \subseteq U$ $$|f(S) - \tolf(S)| \leq \sens/12.$$}
    \STATE{\textbf{Let} $\prec$ denote an arbitrary ordering of $U$.}
    \STATE{\textbf{Let} $\mathcal{I} \leftarrow \{\emptyset\}$}
    \FOR{$x \in U$ (in ascending order under $\prec$)}
    \STATE{$\cI' \leftarrow \emptyset$}
    \FOR{$B \in \cI$}
    \STATE{\textbf{if} $\partial_{x} \tolf(B) > \sens/3$ \textbf{then}
      $\cI' \leftarrow \cI' \cup \{ B \cup \{x\} \}$}
    \ENDFOR
    \STATE{$\cI \leftarrow \cI \cup \cI'$}
    \ENDFOR
    \STATE{\textbf{Let}
$\dom(S) = \{x \in U \mid \partial_x \tolf(S) \leq 2\sens/3\}$
denote the set of elements that have small
marginal value with respect to $S\subseteq U.$}
  \end{algorithmic}

\textbf{Output:} the 
collection of functions $\cG = \{g^B \mid B \in \cI\},$ where 
for $B \in \cI$ we define the
function $g^{B}: 2^{\dom(B)} \to [0,1]$ as $g^{B}(S) = f(S \cup B).$
\label{alg:tolerantdecompose}
\end{algorithm}
  
\begin{lemma}[Lemma~\ref{lem:structure} with tolerance]
\label{lem:tolerantstructure}
Given any submodular function $f\colon2^U\to[0,1]$ and $\sens > 0,$
Algorithm~\ref{alg:tolerantdecompose} makes the following guarantee. There are 
maps $\map,\rest\colon 2^U\to 2^U$ satisfying properties 1-4 of Lemma~\ref{lem:structure} and moreover, can be computed using tolerant queries to $f$ with tolerance $\sens/12$.
\end{lemma}

\begin{proof}
Throughout the proof, we will assume that the oracle always gives the same answer to each query.  Thus the function $\tolf$ defined in Algorithm~\ref{alg:tolerantdecompose} is well defined.  Note
that $\tolf(S)$ need not be submodular even if $f$ is, however, we can assume that we have exact oracle
access to $\tolf(S)$.  Also note that, since we can compute $\partial_x f(S)$ using two queries to $f$, we are guaranteed that for every $S \subseteq U$, and $x \in U$, 
\begin{equation}\label{eq:almostsubmodular}
|\partial_x \tolf(S) - \partial_x f(S)| \leq \sens/6.
\end{equation}

Observe that Algorithm~\ref{alg:tolerantdecompose} differs from Algorithm~\ref{alg:decompose} only in the choice of parameters.  The analysis required to establish the Lemma is also a natural modification of the analysis of Lemma~\ref{lem:structure}, so we will refer the reader to the proof of that Lemma for several details and only call attention to the steps of the proof that require modification.

We will proceed by running through the construction of Lemma~\ref{lem:structure} on $\tolf(S)$ using $\sens/3$ as the error parameter.  Since the argument is a fairly straightforward modification to Lemma~\ref{lem:structure}, we will refer the reader to the proof of that Lemma for several details, and only call attention to the steps of the proof that require modification.

First, we establish a bound on the size of $\cI$
\begin{claim} \label{clm:tolerantsetsize}
  $|\mathcal{I}| \leq |U|^{6/\sens}$
  \end{claim}
  \begin{proof} Let $B \in \mathcal{I}$ be a set, $B = \{x_1, \dots, x_{|B|}\}$.  Let $B_0 = \emptyset$
  and $B_i = \{x_1, \dots, x_i\}$ for $i = 1,\dots,|B|-1$.  Then
  \begin{equation}
    1 \geq f(B) = \sum_{i = 0}^{|B|-1} \partial_{x_{i+1}} f(B_{i}) \geq \sum_{i=0}^{|B|-1} \left( \partial_{x_{i+1}} \tolf(B_{i}) - \sens/6 \right) > |B| \cdot (\sens/3 - \sens/6) = |B| \cdot \sens/6 \mper
  \end{equation}
  Therefore, it must be that $|B| \leq 6/\sens$, and there are at
  most $|U|^{6/\sens}$ such sets over $|U|$ elements.
  \end{proof}
  
  Item~\ref{item:Lipschitz} is shown next.
\begin{claim}[Lipschitz]
   For every $g^{B} \in \cG$, $g^B$ is submodular and $\sup_{x\in \dom(B), S\subseteq \dom(B)} \partial_x g^{B}(S)\le\sens$.
\end{claim}
\begin{proof}
    The proof of submodularity is identical to Claim~\ref{clm:lipschitz}
    
    To establish the Lipschitz property, observe that for every $B \subseteq U$, and every $x \in \dom(B)$, $\partial_x f(B) \leq \partial_x \tolf(B) + \sens/6 \leq \sens$.
\end{proof}

\paragraph{Definition of $\map$ and proof of Item~\ref{item:mapping}.}
In addition to the sets $\dom(B) = \{ x \in U \mid \partial_x \tolf(B) \leq 2\sens/3\}$, we will define the sets $\toldom(B) = \{ x \in U \mid \partial_x \tolf(B) \leq \sens/3\}$, note that for every $B \subseteq U$, $\toldom(B) \subseteq \dom(B)$.  We define the promised mapping $\map(S)$ in the the same manner as in the proof of Lemma~\ref{lem:structure}, but we use $\toldom$ in place of $\dom$ to decide whether or not we select an element $x$ for inclusion in the set $\map(S)$.

Now we establish Property~\ref{item:mapping} via the following claim, analogous to Claim~\ref{clm:complete} in the proof of Lemma~\ref{lem:structure}
  \begin{claim}[Completeness] \label{clm:tolerantcomplete}
  For every $S \subseteq U$, $\map(S) \subseteq S \subseteq \dom(\map(S))$.  Moreover, $g^{\map(S)}(S) = f(S)$.
  \end{claim}
  \begin{proof}
  Let $\path(S) = B_0 \subset B_1 \subset \dots \subset \map(S)$.
 The fact that $\map(S) \subseteq S$ follows as in Claim~\ref{clm:complete}.  To see that $S \subseteq \dom(\map(S))$, assume there exists $x \in S \setminus \dom(\map(S))$.  By submodularity of $f$, and~\eqref{eq:almostsubmodular}, we have 
 $$
 \partial_x \tolf(B_j) \geq \partial_x f(B_j) - \sens/6 \geq \partial_x f(\map(S)) - \sens/6 > \partial_x \tolf(\map(S)) - \sens/3 > \sens/3.
 $$
 Thus, $\partial_x \tolf(B_j) > \sens/3$ for every set $B_j$.  But if $\partial_x f(B_j) > \sens/3$ for every $B_j$ and $x \in S$, then $x \not\in \toldom(B_j)$ for every $B_j$, and it must be that $x \in \map(S)$.  But then $\partial_x f(\map(S)) = 0$, contradicting the fact that $x \not\in \dom(\map(S))$.
  
The fact that $g^{\map(S)}(S) = f(S)$ follows as in the proof of Claim~\ref{clm:complete}.
  \end{proof}

\paragraph{Definition of $\rest$ and proof of Item~\ref{item:unique}.}
We also define the promised mapping $\rest(S)$ in the same manner as in the proof of Lemma~\ref{lem:structure}, but using $\toldom$ in place of $\dom$ to decide whether or not we select an element $x$ for inclusion in the set $\map(S)$.

To establish Property~\ref{item:unique}, we note that the proofs of Claims~\ref{clm:unique1} and~\ref{clm:unique2} do not rely on the submodularity of $f$, therefore they apply as-is to the case where we compute on $\tolf$, even though $\tolf$ is not necessarily submodular.

Property~\ref{item:efficient} also follows as in the proof of Lemma~\ref{lem:structure}.  This completes the proof of the Lemma.
\end{proof}

\newcommand{\comp}[1]{\overline{#1}}

We now present our algorithm for learning monotone submodular functions
over product distributions. For a subset of the universe $V \subseteq
U$, let $\cD_V$ denote the distribution $\cD$ restricted to the
variables in $V$. Note that if $\cD$ is a product distribution, then
$\cD_V$ remains a product distribution and is easy to sample from.

\begin{algorithm}
  \caption{Approximating a monotone submodular function from tolerant queries}
  \textbf{Learn}($f,\alpha,\beta, \mathcal{D}$)
  \begin{algorithmic}
    \STATE Let $\sens = \frac{\alpha^2}{6\log(2/\beta)}.$

    \STATE \textbf{Construct} the collection of functions~$\cG$ returned by Algorithm~\ref{alg:tolerantdecompose} and let $\map, \dom, \rest$ be the associated mappings given by Lemma~\ref{lem:tolerantstructure} with parameter~$\sens.$ 

    \STATE \textbf{Estimate} the value
    $\mu_{g^B}=\E_{S\sim\cD_{\dom(B) \setminus \rest(B)}}[g^B(S)]$ for each
    $g^B\in\cG.$ 

    \STATE \textbf{Output} the data structure~$h$ that consists of the values $\mu_{g^{B}}$ for every $g^{B} \in \cG$ as well as the mapping $\map$.
\end{algorithmic}
\label{alg:approx-submodular0}
\end{algorithm}

\begin{theorem}\label{thm:approx0}
  For any $\alpha, \beta \in (0,1]$, Algorithm~\ref{alg:approx-submodular}
  $(\alpha,\beta)$-approximates any submodular function
  $f\colon2^U\to[0,1]$ under any product distribution $\cD$ in time
  $|U|^{O(\alpha^{-2}\log(1/\beta))}$ using oracle queries to~$f$ of tolerance
  $\alpha^2 / 72 \log(2/\beta)$.
\end{theorem}
\begin{proof}
  For a set $S \subseteq U$, we let $B = \map(S)$ and $g^{B}$ be
  the corresponding submodular function as in Lemma~\ref{lem:tolerantstructure}.  Note that since
  the queries have tolerance $\alpha^2 / 72 \log(1/\beta) \leq \sens/12$, the lemma applies.
  We will analyze the error probability as if the estimates $\mu_{g^{B}}$ were computed using
  exact oracle queries to $f$, and will note that using tolerant queries to $f$ can only introduce an additional
  error of $\alpha^2/72\log(1/\beta) \leq \alpha/6$.
  We claim that, under this condition
  \begin{align}\label{eq:f-to-g0}
    \Pr_{S\sim \cD}\left\{ \left|f(S)-h(S)\right| > 5\alpha/6 \right\}
    &=\Pr_{S \sim \cD}\left\{
      \left|g^{\map(S)}(S)-\mu_{g^{\map(S)}}\right|> 5\alpha/6 \right\} \notag \\
     &= \sum_{g^{B} \in \cG} \Pr_{S\sim \cD}\left\{ B = \map(S) \right\} \cdot \Pr_{S \sim \cD}\left\{ \left|g^{B}(S) - \mu_{g^{B}}\right| >5\alpha/6 \ | \ B = \map(S)\right\} \mper
  \end{align}
  To see this, recall that for every $S \subseteq U$, $g^{\map(S)}(S) = f(S)$.  By Property~\ref{item:unique} of Lemma~\ref{lem:structure}, the condition that $B = \map(S)$ is equivalent to the conditions that $B \subseteq S \subseteq \dom(B)$ and $S \cap \rest(B) = \emptyset$.  Hence,
  \[
   \Pr_{S \sim \cD}\left\{ \left|g^{B}(S) - \mu_{g^{B}}\right| > 5\alpha/6 \ | \ B = \map(S)\right\}
   = \Pr_{S \sim \cD_{\dom(B) \setminus \rest(B)}}\left\{ \left|g^{B}(S) - \mu_{g^{B}}\right| > 5\alpha/6 \right\}\mper
  \]
   Now, applying the concentration inequality for submodular
  functions stated as Corollary~\ref{cor:concentration-submodular}, we
  get
  \begin{equation}
    \Pr_{S\sim\cD_{V_B\setminus T_B}}\left\{\left|g^B(S) - \mu_{g^B}\right|\geq \sens t\right\}
    \leq 2\exp\left(-\frac{t^2}{2(1/\sens + 5/6t)}\right)\mper
  \end{equation}
  Plugging in $t = 5\alpha/6\sens =
  \frac{5 \log(2/\beta)}{\alpha}$ and simplifying we get
  $\Pr_{S\sim \cD_{V(B) \setminus T(B)}}\left\{\left|g^{B}(S) - \mu_{g^{B}}\right| > \alpha\right\} \leq
  \beta\mper$ Combining this with~\eqref{eq:f-to-g0}, the claim
  follows.
\end{proof}

\subsection{Non-monotone Submodular Functions}

For non-monotone functions, we need a more refined argument. Our main structure theorem replaces Property~\ref{item:Lipschitz} in
Lemma~\ref{lem:structure} by the stronger guarantee that
$|\partial_x g(S)|\le\alpha$ for all $g\in\cG$, even for non-monotone submodular functions.
Observe that for a submodular function $f:2^V\rightarrow \bR$, 
the function $\comp{f}\colon 2^V\to\bR$ defined as
$\comp{f}(S)=f(V\backslash S)$ is also submodular; moreover
\begin{equation}
  \label{eq:lip}
  \inf_{x \in V,S \subseteq V}\,\partial_x \comp{f}(S)=-\sup_{x \in V,S \subseteq V}\,\partial_xf(S)\mper
\end{equation}
Given these two facts, we can now prove our main structure theorem.




\begin{algorithm}
\caption{Decomposition for submodular functions from tolerant queries}
\textbf{Input:} Tolerant oracle access to a submodular function $f\colon 2^U\to[0,1]$ with tolerance at most $\sens/12$ and parameter $\sens > 0.$
  \begin{algorithmic}
    \STATE{\textbf{Let} $\tolf$ denote the function specified by tolerant oracle queries to $f$ such that for every $S \subseteq U$ $$|f(S) - \tolf(S)| \leq \sens/12.$$}
    \STATE{\textbf{Let} $\prec$ denote an arbitrary ordering of $U$.}
    \STATE{\textbf{Let} $\cG(f)$ denote the collection of functions returned by Algorithm~\ref{alg:tolerantdecompose} with oracle $f$ and parameter $\sens$, and let $\map_f, \dom_f, \rest_f$ be the associated mappings promised by Lemma~\ref{lem:tolerantstructure}.}
    \FOR{$g^B \in \cG(f)$}
    \STATE{\textbf{Let} $\cG(B)$ be the collection of functions returned by Algorithm~\ref{alg:tolerantdecompose} with oracle $\comp{g}^B$ and parameter $\sens$, and let $\map_B, \dom_B, \rest_B$ be the associated mappings promised by Lemma~\ref{lem:tolerantstructure}.}
    \ENDFOR
        \STATE{\textbf{Let} $\dom(S,T) = \dom_f(S) \cap \dom_S(T)$ denote the set of elements that have small marginal absolute value with respect to $S,T \subseteq U$.}
  \end{algorithmic}

\textbf{Output:} the 
collection of functions $\cG = \bigcup_{g^B \in \cG(f)} \{g^{B,C} = \comp{g}^C \mid g^{C} \in \cG(B)\}$ where $g^{B,C}\colon 2^{\dom(B,C)} \to [0,1]$.
\label{alg:decompose2}
\end{algorithm}

\begin{theorem}
\label{thm:structure}
Given any submodular function $f\colon2^U\to[0,1]$ and $\sens > 0,$
Algorithm~\ref{alg:decompose2} makes the following guarantee. There are 
maps $\map\colon 2^U \to 2^U \times 2^U$ and $\rest\colon 2^U \times 2^U \to 2^U$ such that:
\begin{enumerate}
\item\emph{(Lipschitz)}  \label{item:Lipschitz2} For every $g^{B,C} \in \cG$, $g^{B,C}$ is submodular and satisfies $\sup_{x \in \dom(B,C), S \subseteq \dom(B,C)} |\partial_x g^{B,C}(S) | \leq \sens$.
\item\emph{(Completeness)} \label{item:mapping2} For every 
$S \subseteq U$, $\map(S) \subseteq S \subseteq \dom(\map(S))$
and $g^{\map(S)}(S) = f(S).$
\item\emph{(Uniqueness)} \label{item:unique} 
For every $g^{B,C} \in \cG$, $\map(S) = (B,C)$ if and only if $B,C \subseteq S \subseteq \dom(B,C)$ and $S \cap \rest(B,C) = \emptyset$. 
\item\emph{(Size)} \label{item:efficient} The size of $\cG$ is at most $|\cG| =
|U|^{O(1/\sens)}$.  Moreover, given tolerant oracle access to $f$ with tolerance $\sens/12$, we
compute $\map, \dom, \rest$ in time $|U|^{O(1/\sens)}$.
  \end{enumerate}
\end{theorem}

\begin{proof}
First we show Item~\ref{item:Lipschitz2}
\begin{claim}[Lipschitz]
For every $g^{B,C} \in \cG$, $g^{B,C}$ is submodular and $\sup_{x \in \dom(B,C), S \subseteq \dom(B,C)} |\partial_x g^{B,C}(S) | \leq \sens$.
\end{claim}
\begin{proof}
Submodularity follows directly from Property~\ref{item:Lipschitz} of Lemma~\ref{lem:tolerantstructure}.  The same property of the lemma guarantees that for every $g^B \in \cG(f)$ and $g^C \in \cG(B)$, $\sup_{x \in V_{B}(C), S \subseteq V_{B}(C)} \partial_x g^{C}(S) \leq \sens$.  Moreover, by~\eqref{eq:lip}, $\inf_{x \in V_f(B), S \subseteq V_f(B)} \partial_x \comp{g}^B(S) \geq -\sens$.  Taken together, we obtain $\sup_{x \in V(B,C), S \subseteq V(B,C)} |\partial_x g^{B,C}(S)| \leq \sens$.
\end{proof}

\paragraph{Definition of $\map$ and proof of Item~\ref{item:mapping}.}
Item~\ref{item:mapping2} will follow from the analogous property in Lemma~\ref{lem:tolerantstructure} almost directly.  To construct the mapping $\map(S)$, we want to first compute the appropriate function $g^{B} \in \cG(f)$, using $\map_{f}(S)$ and then find the appropriate function $g^{C} \in \cG(B)$ using $\map_B(S)$.  Thus we can take $\map(S) = (\map_{f}(S), \map_{F_f(S)}(S))$.  By Lemma~\ref{lem:tolerantstructure}, Item~\ref{item:mapping} we have $B \subseteq S \subseteq \dom_{f}(B)$ and $C \subseteq S \subseteq \dom_{B}(C)$, so we conclude $B,C \subseteq S \subseteq \dom(B,C)$.

\paragraph{Definition of $T$ and proof of Item~\ref{item:unique}.}
Item~\ref{item:unique2} will also follow from the analogous property in Lemma~\ref{lem:tolerantstructure}.  By Lemma~\ref{lem:tolerantstructure}, Item~\ref{item:unique}, we have that $\map_f(S) = B$ if and only if $B \subseteq S \subseteq \dom_f(B)$ and $S \cap \rest_f(B) = \emptyset$.  By the same Lemma, we also have that $\map_B(S) = C$ if and only if $C \subseteq S \subseteq \dom_B(C)$ and $S \cap \rest_B(C) = \emptyset$.  So if we define $\rest(B,C) = \rest_f(B) \cup \rest_B(C)$, we can conclude that $\map(S) = (B,C)$ if and only if $B,C \subseteq S \subseteq \dom(B,C)$ and $S \cap \rest(B,C) = \emptyset$.

Now it is clear that $\map(S) = (B,C)$ if $\map_{f}(S) = B$ and $\map_{B}(S) = C$, which by Property~\ref{item:unique} of Lemma~\ref{lem:tolerantstructure} necessitates that $B \subseteq S \subseteq \dom_{f}(B)$, $S \cap \rest_{f}(B) = \emptyset$, $C \subseteq S \subseteq \dom_{B}(S)$, and $S \cap \rest_{B}(C) = \emptyset$.  We have already defined $V(B,C)$ and now we define $\rest(B,C) = \rest_{f}(B) \cup \rest_{B}(C)$.  It is clear now that $\map(S) = (B,C)$ if and only if $B,C \subseteq S \subseteq \dom(B,C)$ and $S \cap \rest(B,C) = \emptyset$.

The size of $\cG$ and running time bounds in Property~\ref{item:efficient2} also follow directly from the analogous property of Lemma~\ref{lem:structure}.  The fact that we can compute the family $\cG$ and the associated mappings $\map,\dom,\rest$ using oracle access to $f$ with tolerance $\sens/12$ follows from the fact that each invocation of Lemma~\ref{lem:structure} can be computed using queries with tolerance $\sens/12$ and from the fact that Algorithm~\ref{alg:decompose2} only queries $f$ in order to invoke Lemma~\ref{lem:tolerantstructure}.  This completes the proof of the Theorem.
\end{proof}


We now present our algorithm for learning arbitrary submodular functions
over product distributions. For a subset of the universe $V \subseteq C$, let $\cD_V$ denote the distribution $\cD$ restricted to the variables in $V$. Note that if $\cD$ is a product distribution, then $\cD_V$ remains a product distribution and is easy to sample from.
\begin{algorithm}
  \caption{Approximating a non-monotone submodular function}
  \textbf{Learn}($f,\alpha,\beta, \mathcal{D}$)
  \begin{algorithmic}
    \STATE Let $\sens = \frac{\alpha^2}{6\log(2/\beta)}.$

    \STATE \textbf{Construct} the collection of functions~$\cG$ and the associated mappings $\map, \dom, \rest$ given by
    Theorem~\ref{thm:structure} with parameter~$\sens.$ 

    \STATE \textbf{Estimate} the value
    $\mu_{g^{B,C}}=\E_{S \sim \cD_{V(B,C) \setminus T(B,C)}}[g^{B,C}(S)]$ for each
    $g^{B,C}\in\cG.$ 

    \STATE \textbf{Output} the data structure~$h$ that consists of the values $\mu_{g^{B,C}}$ for every $g^{B,C} \in \cG$ as well as the mapping $\map$.
\end{algorithmic}
\label{alg:approx-submodular}
\end{algorithm}
To avoid notational clutter, throughout this section we will not consider the details of how we construct our estimate $\mu_g$. However, it is an easy observation that this quantity can be estimated to a sufficiently high degree of accuracy using a small number of random samples.
\begin{theorem}\label{thm:approx1}
  For any $\alpha, \beta \in (0,1]$, Algorithm~\ref{alg:approx-submodular}
  $(\alpha,\beta)$-approximates any submodular function $f\colon2^U\to[0,1]$
  under any product distribution in time $|U|^{O(\alpha^{-2}\log(1/\beta))}$
using oracle queries to~$f$ of tolerance $\alpha^2/72\log(1/\beta)$
\end{theorem}
\begin{proof}
  For a set $S \subseteq U$, we let $(B,C) = \map(S)$ and $g^{B,C}$ be
  the corresponding submodular function as in Theorem~\ref{thm:structure}.
  Note that since
  the queries have tolerance $\alpha^2 / 72 \log(1/\beta) \leq \sens/12$, the lemma applies.
  We will analyze the error probability as if the estimates $\mu_{g^{B}}$ were computed using
  exact oracle queries to $f$, and will note that using tolerant queries to $f$ can only introduce an additional
  error of $\alpha^2/72\log(1/\beta) \leq \alpha/6$.
  We claim that, under this condition
  We claim that
  \begin{align}\label{eq:f-to-g1}
    \Pr_{S\sim \cD}\left\{ \left|f(S)-h(S)\right| > 5\alpha/6 \right\}
    &=\Pr_{S \sim \cD}\left\{
      \left|g^{\map(S)}(S)-\mu_{g^{\map(S)}}\right|> 5\alpha/6 \right\} \notag \\
     &= \sum_{g^{B,C} \in \cG} \Pr_{S\sim \cD}\left\{ (B,C) = \map(S) \right\} \cdot \Pr_{S \sim \cD}\left\{ \left|g^{B,C}(S) - \mu_{g^{B,C}}\right| > 5\alpha/6 \ | \ (B,C) = \map(S)\right\} \mper
  \end{align}
  To see this, recall that for every $S \subseteq U$, $g^{\map(S)}(S) = f(S)$.  By Property~\ref{item:unique} of Lemma~\ref{lem:structure}, the condition that $B = \map(S)$ is equivalent to the conditions that $B,C \subseteq S \subseteq \dom(B,C)$ and $S \cap \rest(B,C) = \emptyset$.  Hence,
  \[
   \Pr_{S \sim \cD}\left\{ \left|g^{B,C}(S) - \mu_{g^{B,C}}\right| > 5\alpha/6 \ | \ (B,C) = \map(S)\right\}
   = \Pr_{S \sim \cD_{\dom(B,C) \setminus \rest(B,C)}}\left\{ \left|g^{B}(S) - \mu_{g^{B}}\right| > 5\alpha/6 \right\}\mper
  \]
   Now, applying the concentration inequality for submodular
  functions stated as Corollary~\ref{cor:concentration-submodular}, we
  get
  \begin{equation}
    \Pr_{S\sim\cD_{V(B,C) \setminus T(B,C)}}\left\{\left|g^{B,C}(S) - \mu_{g^{B,C}}\right|\geq \sens t\right\}
    \leq 2\exp\left(-\frac{t^2}{2(1/\sens + 5t/6)}\right)\mper
  \end{equation}
  Plugging in $t = 5\alpha/6\sens =
  \frac{5\log(2/\beta)}{\alpha}$ and simplifying we get
  $\Pr_{S\sim \cD_{V(B,C) \setminus T(B,C)}}\left\{\left|g^{B,C}(S) - \mu_{g^{B,C}}\right| > \alpha\right\} \leq
  \beta\mper$ Combining this with Equation~\eqref{eq:f-to-g1}, the claim
  follows.
\end{proof}

\forsubmit{We can extend these ideas to approximating submodular functions
  where answers are approximate, more general submodular-like functions,
  and some non-product distributions. Details appear in Appendix~\ref{app:sec-submod}.}

\section{Applications to privacy-preserving query release}
\label{sec:conjunctions}
In this section, we show how to apply our algorithm from
Section~\ref{sec:submodular} to the problem of releasing monotone conjunctions over a
boolean database. 
In Section~\ref{app:graphcut}, we also show how our mechanism can be applied
to release the \emph{cut function} of an arbitrary graph.

Let us now begin with the monotone disjunctions. We will then extend the
result to monotone conjunctions.  Given our previous results, we only need to
argue that monotone disjunctions can be described by a submodular function. 
Indeed, every element $S\in\{0,1\}^d$ naturally corresponds to a monotone
Boolean disjunction $d_S\colon\{0,1\}^d\to\{0,1\}$ by putting 
\[ 
d_S(x) \defeq \bigvee_{i \colon S(i) = 1}x_i\mper 
\] 
Note that in contrast to
Section~\ref{sec:submodular} here we use $x$ to denote an element
of~$\{0,1\}^d.$
Let $F_{\textrm{Disj}}: \{0,1\}^d \rightarrow [0,1]$ be the function such
that $F_{\textrm{Disj}}(S) = d_S(D)$.  It is easy to show that
$F_{\textrm{Disj}}(S)$ is a monotone submodular function.
\begin{lemma}
$F_{\textrm{Disj}}$ is a monotone submodular function.
\end{lemma}
\begin{proof}
Let $X_i^+$ denote the set of elements $x \in D$ such that $x_i = 1$, and let
$X_i^-$ denote the set of elements $x \in D$ such that $x_i = 0$. Consider the
set system $U = \{X_i^+, X_i^-\}_{i=1}^d$ over the universe of elements $x \in
D$. Then there is a natural bijection between $F_{\textrm{Disj}}(D)$ and the
set coverage function $\textrm{Cov}:2^U\rightarrow [0,|D|]$ defined to be
$\textrm{Cov}(S) = |\bigcup_{X \in U} X|$, which is a monotone submodular
function.  
\end{proof}
We therefore obtain the following corollary directly by combining
Theorem~\ref{thm:approx0} with Proposition~\ref{prop:sq2dp}.
\begin{corollary}
\label{cor:disjunctions}
Let $\alpha,\beta,\epsilon>0.$ There is an
$\epsilon$-differentially private algorithm that $(\alpha,\beta)$-releases the
set of monotone Boolean disjunctions over any product distribution 
in time $d^{t(\alpha,\beta)}$ for
any data set of size~$|D|\ge d^{t(\alpha,\beta)}/\epsilon$ where
$t(\alpha,\beta)=O(\alpha^{-2}\log(1/\beta)).$
\end{corollary}

For completeness, we will present the algorithm for privately releasing monotone disjunctions over a product
distribution $\cD$ for a data set $D$, though we will rely on Corollary~\ref{cor:disjunctions} for the formal analysis.
\begin{algorithm}
 \caption{Privately Releasing Monotone Disjunctions}
  \textbf{Release}($D,\alpha,\beta, \epsilon, \mathcal{D}$)
\begin{algorithmic}
    \STATE \textbf{Simulate} the oracle queries $F_{Disj}(S)$ by answering with $d_{S}(D) + Lap(t(\alpha, \beta)/\epsilon |D|)$.

    \STATE Let $\sens = \frac{\alpha^2}{6\log(2/\beta)}.$

    \STATE \textbf{Construct} the collection of functions~$\cG$ and the associated mappings $\map, \dom, \rest$ given by Lemma~\ref{lem:tolerantstructure} on the function $F_{Disj}$ with parameter~$\sens.$

    \STATE \textbf{Estimate} the value
    $\mu_{g^{B}}=\E_{S \sim \cD_{V(B) \setminus T(B)}}[g^{B}(S)]$ for each
    $g^{B}\in\cG.$ 

    \STATE \textbf{Output} the data structure~$h$ that consists of the estimated values $\mu_{g^{B}}$ for every $g^{B} \in \cG$ as well as the mapping $\map$.  To evaluate any monotone disjunction query $d_{S}(D)$, compute $\mu_{g^{\map(S)}}$.
\end{algorithmic}
\end{algorithm} 

We will next see that this corollary directly transfers to monotone
conjunctions.
A monotone Boolean conjunction $c_S: \{0,1\}^d \rightarrow \{0,1\}$ is
defined as 
\[
c_S(x) \defeq \bigwedge_{i\in S}x_i
= 1-\bigvee_{i\in S} (1-x_i)\mper
\]
Given the last equation, it is clear that in order to release 
conjunctions over some distribution, it is sufficient to release 
disjunctions over the same distribution after replacing every data item $x\in
D$ by its negation $\bar x,$ i.e., $\bar x_i=1-x_i.$ Hence,
Corollary~\ref{cor:disjunctions} extends directly to monotone conjunctions.

\paragraph{Extension to width~$w.$}
Note that the uniform distribution on disjunctions of width~$w$ is
not a product distribution, which is what we require to apply
Theorem~\ref{thm:approx1} directly. However, in Lemma~\ref{lem:product2slice}
we show that for monotone submodular functions (such as
$F_{\textrm{Disj}}^{D}$) the concentration of measure property required in the
proof Theorem~\ref{thm:approx1} is still satisfied.  Of course, we can
instantiate the theorem for every $w\in\{1,\dots,k\}$ to obtain a statement
for conjunctions of any width.

Indeed, given a monotone submodular function $f\colon 2^U\rightarrow
\bR$, let $S \in 2^U$ be the random variable where for every $x\in U,$
independently $x\in S$ with probability $w/d$ and $x\not\in S$ with
probability $1-w/d.$ On the other hand, let $T\in
2^U$ denote the uniform distribution over strings in $2^U$ of weight
$w.$ The following lemma is due to Balcan and Harvey~\cite{BH10}.
\begin{lemma}
\label{lem:product2slice}
Assume $f:2^U\rightarrow \bR$ is monotone function,
and $S$ and $T$ are chosen at random as above. Then,
\begin{equation}
\label{direc:1}
\Pr[f(T) \geq \tau] \leq 2\Pr[f(S) \geq \tau]
\end{equation}
\begin{equation}
\label{direc:2}
\Pr[f(T) \leq \tau] \leq 2\Pr[f(S) \leq \tau]
\end{equation}
\end{lemma}

\begin{remark}
Throughout this section we focus on the case of \emph{monotone} disjunctions
and conjunctions. Our algorithm can be extended to non-monotone
conjunctions/disjunctions as well. However, this turns out to be less
interesting than the monotone case. Indeed, a random non-monotone conjunction of width $w$ is false on any fixed data item with probability $2^{-w}$, thus when $w \geq \log(1/\alpha)$, the constant function $0$ is a good approximation to $F_{Disj}$ on a random non-monotone conjunction of width $w$.  We therefore omit the non-monotone case from our presentation.
\end{remark}

\subsection{Releasing the cut function of a graph}
\label{app:graphcut}
Consider a graph $G = (V, E)$ in which the edge-set represents the private
database (We assume here that each individual is associated with a single edge
in $G$. The following discussion generalizes to the case in which individuals
may be associated with multiple edges, with a corresponding increase in
sensitivity). The \emph{cut function} associated with $G$ is
$f_G:2^V\rightarrow [0,1]$, defined as: $$f_G(S) =
\frac{1}{|V|^2}\cdot\left|\{(u,v) \in E : u \in S, v \not \in S\}\right|$$ We
observe that the graph cut function encodes a collection of counting queries
over the database $E$ and so has sensitivity $1/|V|^2$.
\begin{fact} For any
graph $G$, $f_G$ is submodular.  \end{fact}
\begin{lemma}
\label{lem:Gsizeforcuts}
The decomposition from Theorem~\ref{thm:structure} constructs a collection of functions $\cG$ of size $|\cG| \leq 2^{2/\alpha}$.
\end{lemma}
\begin{proof}
Let $u \in V$, and $S \subset V$ such that $|\partial_uf_G(S)| \geq \alpha$. It must be that the degree of $u$ in $G$ is at least $\alpha\cdot |E|$. But there can be at most $2/\alpha$ such high-influence vertices, and therefore at most $2^{2/\alpha}$ subsets of high influence vertices.
\end{proof}

\begin{corollary}
Algorithm \ref{alg:approx-submodular} can be used to privately $(\alpha,\beta)$-release the cut function on any graph over any product distribution in time $t(\alpha, \beta, \epsilon)$ for any database of size $|D| \geq t(\alpha, \beta, \epsilon)$, while preserving $\epsilon$-differential privacy, where:
$$t(\alpha, \beta, \epsilon) = \frac{2^{O(\alpha^{-2}\log(1/\beta))}}{\epsilon}$$
\end{corollary}
\begin{proof}
This follows directly from  a simple modification of Theorem \ref{thm:approx1}, by applying Lemma~\ref{lem:Gsizeforcuts} and plugging in the size of the decomposition $\cG$. The algorithm can then be made privacy preserving by applying proposition \ref{prop:sq2dp}.
\end{proof}

\section{Equivalence between agnostic learning and query release}
\label{sec:agnostic}

\newcommand{\cB}{{\cal B}}
\newcommand{\stat}{\mathrm{STAT}}

In this section we show an information-theoretic equivalence between
\emph{agnostic learning} and \emph{query release} in the statistical queries
model.  In particular, given an agnostic learning algorithm for a specific
concept class we construct a query release algorithm for the same concept
class.

Consider a distribution $A$ over $X \times \{0,1\}$ and a concept class $C$.
An \emph{agnostic learning} algorithm (in the strong sense) finds the concept
$q \in C$ that approximately maximizes $\Pr_{(x,b) \sim A} \left\{q(x) =
b\right\}$ to within an additive error of $\alpha$.  Our reduction from query
release to agnostic learning actually holds even for \emph{weak agnostic
learning}.  A weak agnostic learner is not required to maximize $\Pr_{(x,b)
\sim A} \left\{q(x) = b\right\}$, but only to find a sufficiently good
predicate $q$ provided that one exists.

We use $\stat_\tau(A)$ to denote the \emph{statistical query oracle for distribution $A$} that takes as input a predicate $q: X \to \{0,1\}$ and returns a value $v$ such that $| v - \mathbb{E}_{x \sim A}[q(x)] |$

\begin{definition}[Weak Agnostic SQ-Learning] \label{def:swal}
Let $C$ be a concept class and $\gamma, \tau > 0$ and $0 < \beta < \alpha \leq
1/2$.  An algorithm $\cA$ with oracle access to $\stat_\tau(A)$ is an
\emph{$(\alpha, \beta, \gamma, \tau)$-weak agnostic learner for $C$} if for
every distribution $A$ such that there exists $q^* \in C$ satisfying
$\Pr_{(x,b) \sim A} \left\{q^*(x) = b\right\}\ge \nfrac12 + \alpha \mcom$
$\cA(A)$ outputs a predicate $q: X \rightarrow \{0,1\}$ such that
$\Pr_{(x,b) \sim A} \left\{q(x) = b\right\}\ge \nfrac12 + \beta \mcom$
with probability at least $1-\gamma.$
\end{definition}
Note that if we can agnostically learn~$C$ in the strong sense from queries of
tolerance~$\tau$ to within additive error $\alpha - \beta$ with probability
$1-\gamma,$ then there is also an $(\alpha,\beta,\gamma,\tau)$-weak agnostic
learner.

We are now ready to state the main result of this section, which shows that a
weak agnostic SQ-learner for any concept class is sufficient to release the
same concept class in the SQ model.

\begin{theorem} \label{thm:agnostic2release}
Let $C$ be a concept class.
Let $\cA$ be an algorithm that $(\alpha/2, \beta, \gamma, \tau)$ weak agnostic-SQ
learns $C$ with $\tau \leq \beta/8$.  Then there exists an algorithm $\cB$ that invokes
$\cA$ at most $T = 8 \log |X| / \beta^2$ times and $(\alpha, 0)$-releases $C$ with probability at least $1 - T\gamma$.
\end{theorem}

\begin{algorithm}
\caption{Multiplicative weights update}\label{alg:mult-weights}
\noindent Let $D_0$ denote the uniform distribution over~$X.$

\noindent \textbf{For} $t=1,...,T = \lceil8\log|X|/\beta^2\rceil+1$:
\begin{algorithmic}
\STATE Consider the distributions
\begin{equation*}
A^+_t = \nfrac12(D,1) + \nfrac12(D_{t-1},0) \qquad\qquad
A^-_t = \nfrac12(D,0) + \nfrac12(D_{t-1},1) \mper
\end{equation*}
Let $q^+_t=\cA(A^+_t)$ and $q^-_t=\cA(A^-_t)$.
Let $v^+_t$ be the value returned by $\stat_\tau(A^+_t)$ on the query $q^+_t$ and $v^-_t$ be the value returned by $\stat_\tau(A^-_t)$ on the query $q^-_t$.  Let $v_t = \max\{v^+_t, v^-_t\} - 1/2$ and $q_t$ be the corresponding query.
\STATE \textbf{If:}
\begin{equation}\label{eq:termination}
v_t \le \frac\beta2 - \tau \mcom
\end{equation}
proceed to ``output'' step.
\STATE \textbf{Update:} Let $D_t$ be the distribution obtained from $D_{t-1}$ using a multiplicative weights update step with penalty function induced by $q_t$ and penalty parameter $\eta = \beta/2$ as follows:
\[
D'_{t}(x) = \exp(\eta q_t(x))\cdot D_{t-1}(x)
\]
\[
D_{t}(x) = \frac{ D'_{t} }{ \sum_{x \in X} D'_{t}(x)}
\]
\end{algorithmic}
\textbf{Output} $a_c=\E_{x\sim D_T}c(x)$ for each $c\in C.$
\end{algorithm}
The proof strategy is as follows. We will start from $D_0$ being the uniform
distribution over $X.$ We will then construct a short sequence of distributions
$D_1,D_2,\dots,D_T$ such that no concept in $C$ can distinguish between $D$
and $D_T$ up to bias $\alpha.$  Each distribution $D_t$ is obtained from the
previous one using a multiplicative weights approach as in~\cite{HardtRo10}
and with the help of the learning algorithm that's given in the assumption of
the theorem.
Intuitively, at every step we use the agnostic learner to give us the
predicate $q_t\in C$ which distinguishes between $D_t$ and $D.$ In
order to accomplish this we feed the agnostic learner with the distribution
$A_t$ that labels elements sampled from $D$ by~$1$ and elements sampled from
$D_t$ by~$0.$ For a technical reason we also need to consider the distribution
with $0$ and $1$ flipped. Once we obtained~$q_t$ we can use it as a penalty
function in the update rule of the multiplicative weights method. This has the
effect of bringing $D$ and $D_t$ closer in relative entropy. A typical
potential argument then bounds the number of update steps that can occur
before we reach a distribution $D_t$ for which no good distinguisher in~$C$
exists.
\forreals{
\subsection{Proof of Theorem~\ref{thm:agnostic2release}}

\newcommand{\opt}{\mathrm{opt}}

\begin{proof}
We start by relating the probability that $q_t$ predicts $b$ from $x$ on the distribution $A^+_t$ to the difference in expectation of $q_t$ on $D$ and $D_{t-1}$.

\begin{lemma}
\label{lem:prob2ex}
For any $q\colon X \to \{0,1\}$,
\begin{equation} \label{eq:prob2errpos}
\Pr_{(x,b) \sim A^+_t} \{q(x) = b\} - \frac12 = \frac12 \left( \E_{x\sim D}q(x) - \E_{x\sim D_{t-1}}q(x) \right)
\end{equation}
\end{lemma}

\begin{proof}
If $q_t = q^+_t$ then
\begin{align*}
\Pr_{(x,b)\sim A^+_t}\{q(x)=b\}
& = \frac12\Pr_{x \sim D}\{q(x)=1\}+\frac12\Pr_{x \sim D_{t-1}}\{q(x)=0\} \\
& = \frac12\E_{x \sim D} \left[q(x)\right]
 + \frac12\E_{x \sim D_{t-1}}\left[1-q(x)\right] \\
& = \frac12+\frac12\left(\E_{x\sim D}q(x) - \E_{x\sim D_{t-1}}q(x)\right)
\end{align*}
Note that $\Pr_{(x,b) \sim A^-_t}\{q(x) = b\} = 1 - \Pr_{(x,b) \sim A^-_t}\{q(x) = (1-b)\} = 1 - \Pr_{(x,b) \sim A^+_t}\{ q(x) = b \}$, so if $q_t = q^-_t$ then
\begin{align*}
\Pr_{(x,b) \sim A^+_t}\{q(x) = b\} = 1 - \Pr_{(x,b) \sim A^-_t}\{q(x) = b\}
&= 1 - \left(\frac12 - \frac12\left(\E_{x\sim D}q(x) - \E_{x\sim D_{t-1}}q(x)\right) \right) \\
&= \frac12 + \frac12\left(\E_{x\sim D}q(x) - \E_{x\sim D_{t-1}}q(x)\right)
\end{align*}
\end{proof}

The rest of the proof closely follows~\cite{HardtRo10}.  For two distributions $P, Q$ on a universe $X$ we define
the \emph{relative entropy} to be $\KL(P||Q) = \sum_{x \in X} P(x) \log (P(x) / Q(x))$.
We consider the potential
\[
\Psi_t = \KL(D||D_t)\mper
\]
\begin{fact}\label{fact:nonneg}
$\Psi_t\ge0$
\end{fact}
\begin{fact}\label{fact:startvalue}
$\Psi_0\le\log|X|$
\end{fact}
We will argue that in every step the potential drops by at least $\beta^2/4$
Hence, we know that there can be at most $4\log|X|/\alpha^2$
steps before we reach a distribution that satisfies~(\ref{eq:termination}).

The next lemma gives a lower bound on the potential drop in terms of the
concept, $q_t$, returned by the learning algorithm at time~$t$. Recall, that
$\eta$ (used below) is the penalty parameter used in the multiplicative
weights update rule.

\begin{lemma}[\cite{HardtRo10}]
\label{lem:drop}
\begin{equation}
\Psi_{t-1}-\Psi_t\ge
\eta\left|\E_{x \sim D} q_t(x)-\E_{x \sim D_{t-1}} q_t(x)\right|-\eta^2
\end{equation}
\end{lemma}

Let
\[
\opt_t = \sup_{q \in C} \left| \Pr_{(x,b) \sim A^+_t}\left\{ q(x) = b \right\} - \frac12 \right| \mper
\]
Note that $\Pr_{(x,b) \sim A^-_t} \{ q(x) = b \} = 1 - \Pr_{(x,b) \sim A^+_t} \{ \neg q(x) = b \}$.  For the remainder of the proof we treat the two cases symmetrically and only look at how far from $1/2$ these probabilities are.
The next lemma shows
that either $\opt_t$ is large or else we are done in the sense that $D_t$ is
indistinguishable from $D$ for any concept from $C.$

\begin{lemma}\label{lem:opt}
Let $\alpha>0.$ Suppose
\[
\opt_t \le \frac{\alpha}{2} \mper
\]
Then, for all $q \in C,$
\begin{equation}
\left|
\E_{x\sim D} q(x)
-\E_{x\sim D_t} q_t(x)
\right|
\le \alpha
\end{equation}
\end{lemma}

\begin{proof}
From Lemma~\ref{lem:prob2ex} we have that for every $q \in C$
\begin{equation*}
\frac{\alpha}{2} \geq \opt_t \geq \Pr_{(x,b)\sim A^+_t}\{q(x)=b\} - \frac12 = \frac12 \left(
\E_{x\sim D} q(x)
-\E_{x\sim D_t} q_t(x)
\right)
\end{equation*}
Thus $\alpha \geq \left(
\E_{x\sim D} q(x)
-\E_{x\sim D_t} q_t(x)
\right)$.
Similarly,
\begin{equation*}
\frac{\alpha}{2} \geq \opt_t \geq \Pr_{(x,b)\sim A^-_t}\{q(x)=b\} - \frac12 = \frac12 \left(
\E_{x\sim D_{t}} q(x)
-\E_{x\sim D} q_t(x)
\right)
\end{equation*}
Thus $- \alpha \leq \left(
\E_{x\sim D} q(x)
-\E_{x\sim D_t} q_t(x)
\right)$.  So we conclude $\alpha \geq \left|
\E_{x\sim D} q(x)
-\E_{x\sim D_t} q_t(x)
\right|.$
\end{proof}

We can now finish the proof of Theorem~\ref{thm:agnostic2release}. By our
assumption, we have that so long as $\opt_t \geq \alpha/2$ the algorithm $\cA$
produces a concept $q_t$ such that with probability $1-\gamma$
\begin{equation}\label{eq:hyp2opt}
\left| \Pr_{(x,b) \sim A^+_t}\{q_t(x)=b\} - \frac12 \right| \ge \beta \mper
\end{equation}
For the remainder of the proof we assume that our algorithm returns a concept satisfying
Equation~\eqref{eq:hyp2opt} in every stage for which $\opt_t \geq  \alpha/2$.  By a union bound
over the stages of the algorithm, this event occurs with probability at least $1-T\gamma$.

Assuming Equation~\eqref{eq:termination} is not satisfied we have that
\[
\frac{\beta}{4} \leq \frac{\beta}{2} - 2\tau \leq v_t - \tau \leq \left| \Pr_{A^+_t}\{q_t(x) = b\} \right| \mper
\]
The leftmost inequality follows because $\tau \leq \beta/8$.
We then get
\begin{align*}
\Psi_{t-1}-\Psi_t
&\ge
\eta\left|\E_D q_t(x)-\E_{D_{t-1}} q_t(x)\right|-\eta^2 \tag{Lemma~\ref{lem:drop}}\\
&\ge
\eta\left|4\Pr_{A_t}\{q_t(x)=b\}-2\right|-\eta^2 \tag{Lemma~\ref{lem:prob2ex}}\\
&\ge \eta \cdot \beta -\eta^2
\tag{Equation~\ref{eq:termination} not satisfied}\\
&\ge \frac{\beta^2}{2}-\frac{\beta^2}{4}
\tag{$\eta=\beta/2$}\\
&= \frac{\beta^2}{4}
\end{align*}

Hence, if we put $T\ge4\log|X|/\beta^2,$ we must reach a distribution that satisfies~(\ref{eq:termination}).  But at that point, call it $t$, the subroutine $\cA$ outputs a concept $q_{t}$ such that
\[
\left| \Pr_{(x,b) \sim A^+_t}(q_{t}(x) = b) - \frac12 \right| \leq v_t + \tau < \frac{\beta}{2} + \tau < \beta
\]
In this case, by our assumption that Equation~\ref{eq:hyp2opt} is satisfied
whenever $\mathrm{opt}_t \geq 1/2 + \alpha/2$, we conclude that $\mathrm{opt}_t < 1/2 + \alpha/2$.
By Lemma~\ref{lem:opt}, we get
\begin{equation*}
\sup_{q \in C}\left|
\E_{x\sim D} q(x)
-\E_{x\sim D_t} q_t(x)
\right|\le\alpha\mper
\end{equation*}
But this is what we wanted to show, since it means that our output on all
concepts in $C$ will be accurate up to error~$\alpha.$
\end{proof}

We remark that for clarity, we let the failure probability of the release algorithm grow linearly in the number of calls we made to the learning algorithm (by the union bound). However, this is not necessary: we could have driven down the probability of error in each stage by independent repetition of the agnostic learner.
}

This equivalence between release and agnostic learning also can easily be seen
to hold in the reverse direction as well. 

\begin{theorem} \label{thm:release2agnostic}
Let $C$ be a concept class. If there exists an algorithm $\cB$ that $(\alpha, 0)$-releases $C$ with probability
$1-\gamma$ and accesses the database using at most $k$ oracle accesses to $\textrm{STAT}_\tau(A)$, then there is an algorithm that makes $2k$ queries to $\stat_\tau(A)$ and agnostically learns $C$ in the strong sense with accuracy $2\alpha$ with probability at least $1-2\gamma$.
\end{theorem}
\begin{proof}
Let $Y$ denote the set of examples with label $1$, and let $N$ denote the set
of examples with label $0$.  We use $\stat_\tau(A)$ to simulate oracles
$\stat_\tau(Y)$ and $\stat_\tau(N)$ that condition the queried concept on the
label.  That is, $\stat_\tau(Y)$, when invoked on concept $q$, returns an
approximation to $\Pr_{x \sim A}\{q(x) = 1 \land (x \in Y) \}$ and
$\stat_\tau(N)$ returns an approximation to $\Pr_{x \sim A}\{q(x) = 1 \land (x
\in Y)]$.  We can simulate a query to either oracle using only one query to
$\stat_\tau(A)$.

Run $\cB(Y)$ to obtain answers $a^Y_1,\ldots,a^Y_{|C|}$, and run $\cB(N)$ to
obtain answers $a^N_1,\ldots,a^N_{|C|}$. Note that this takes at most $2k$
oracle queries, using the simulation described above, by our assumption on
$\cB$.  By the union bound, except with probability $2\gamma$, we have for all
$q_i \in C$: $|q_i(Y) - a^Y_i| \leq \alpha$ and $|q_i(B) - a^N_i| \leq
\alpha$. Let $q^* = \arg\max_{q_i \in C}(a^Y_i - a^N_i)$. Observe that $q^*(D)
\geq \max_{q\in C}q(D) - 2\alpha$, and so we have agnostically learned $C$ up
to error $2\alpha$.  
\end{proof}

Feldman proves that even monotone conjunctions cannot be agnostically learned to subconstant error with polynomially many SQ queries:

\begin{theorem}[\cite{Feldman10}]
Let $C$ be the class of monotone conjunctions. Let $k(d)$ be any polynomial in $d$, the dimension of the data space. There is no algorithm $\cA$ which agnostically learns $C$ to error $o(1)$ using $k(D)$ queries to $\textrm{STAT}_{1/k(d)}$.
\end{theorem}

\begin{corollary}
\label{cor:lowerbound}
For any polynomial in $d$, $k(d)$, no algorithm that makes $k(d)$ statistical queries to a database of size $k(d)$ can release the class of monotone conjunctions to error $o(1)$.
\end{corollary}

Note that formally, Corollary~\ref{cor:lowerbound} only precludes algorithms
which release the approximately correct answers to \emph{every} monotone
conjunction, whereas our algorithm is allowed to make arbitrary errors on a
small fraction of conjunctions.
\begin{remark}
It can be shown that the lower bound from Corollary~\ref{cor:lowerbound} 
in fact does \emph{not} 
hold when the accuracy requirement is relaxed so that the algorithm may err arbitrarily on $1$\% of all 
the conjunctions. Indeed, there is an inefficient algorithm (runtime
$\poly(2^d)$) that makes $\poly(d)$ statistical queries and releases random
conjunctions up to a small additive error. 
The algorithm roughly proceeds by running multiplicative weights
privately (as in~\cite{HardtRo10} or above) 
while sampling, say, $1000$ random 
conjunctions at every step and checking if any of them have large error. If
so, an update occurs. We omit the formal description and analysis of the
algorithm.
\end{remark}

We also remark that the proofs of Theorems~\ref{thm:agnostic2release}
and~\ref{thm:release2agnostic} are not particular to the statistical queries
model: we showed generically that it is possible to solve the query release
problem using a small number of black-box calls to a learning algorithm,
\emph{without accessing the database except through the learning algorithm}.
This has interesting implications for any class of algorithms that may make
only restricted access to the database. For example, this also proves that if
it is possible to agnostically learn some concept class $C$ while preserving
$\epsilon$-differential privacy (even using algorithms that do not fit into
the SQ model), then it is possible to release the same class while preserving
$T\epsilon \approx \log |X|\epsilon$-differential privacy.

\section*{Acknowledgments}
We would like to thank Guy Rothblum and Salil Vadhan for many insightful discussions, and Nina Balcan and Nick Harvey for pointing out key distinctions between our algorithmic model and that of \cite{BH10}.

\bibliographystyle{alpha}
\bibliography{localconjunctions}
\appendix

\end{document}